\documentclass[pra,twocolumn,groupedaddress,superscriptaddress,amsmath,amssymb,a4paper,showkeys]{revtex4}
\usepackage{geometry}
\usepackage{graphicx}
\usepackage{mathrsfs}
\usepackage{amssymb}
\usepackage{amsmath}
\usepackage{amsthm}
\usepackage{amsbsy}
\usepackage{bm}
\usepackage{hyperref}
\usepackage[T1]{fontenc}
\usepackage{color}
\newtheorem{proposition}{Proposition}

\theoremstyle{definition}

\newtheorem{example}{Example}

\newcommand{\bra}[1]{\langle #1|}
\newcommand{\ket}[1]{| #1 \rangle }
\newcommand{\ip}[2]{{\langle #1|}{ #2 \rangle }}

\begin{document}

\title{Spin polarization-scaling quantum maps and channels}

\author{Sergey N. Filippov}

\affiliation{Moscow Institute of Physics and Technology,
Institutskii Per. 9, Dolgoprudny, Moscow Region 141700, Russia}

\affiliation{Valiev Institute of Physics and Technology of Russian
Academy of Sciences, Nakhimovskii Pr. 34, Moscow 117218, Russia}

\affiliation{Steklov Mathematical Institute of Russian Academy of
Sciences, Gubkina St. 8, Moscow 119991, Russia}

\author{Kamil Yu. Magadov}

\thanks{Deceased July 6, 2017.}

\affiliation{Moscow Institute of Physics and Technology,
Institutskii Per. 9, Dolgoprudny, Moscow Region 141700, Russia}

\begin{abstract}
We introduce a spin polarization-scaling map for spin-$j$
particles, whose physical meaning is the decrease of spin
polarization along three mutually orthogonal axes. We find
conditions on three scaling parameters under which the map is
positive, completely positive, entanglement breaking,
2-tensor-stable positive, and 2-locally entanglement annihilating.
The results are specified for maps on spin-$1$ particles. The
difference from the case of spin-$\frac{1}{2}$ particles is
emphasized.
\end{abstract}


\maketitle

\section{\label{section-introduction} Introduction}

Quantum states of a spin-$j$ particle are described by
$(2j+1)\times (2j+1)$ density matrices $\varrho \in \mathcal{S}
(\mathcal{H})$ satisfying the properties $\varrho^{\dag} =
\varrho$, ${\rm tr}[\varrho] = 1$, and $\bra{\varphi} \varrho
\ket{\varphi} \geqslant 0$ for all $\ket{\varphi}\in \mathcal{H}$,
${\rm dim}\mathcal{H} = 2j+1$. Taking into account the
normalization condition, the density matrix $\varrho$ is defined
by $(2j+1)^2-1$ real parameters, which are usually treated as
components of the generalized Bloch
vector~\cite{petz-2009,checinska-2009,karimipour-2011,byrd-2011,goyal-2016}.
However, many physical phenomena can be explained and visualized
via a spin polarization vector ${\bf p} \in \mathbb{R}$ with
components $p_i = {\rm tr}[\varrho J_i]$, where $J_1,J_2,J_3$ are
usual $(2j+1)$-dimensional representations of angular momentum
operators (see, e.g.,~\cite{varshalovich}). Angular momentum
operators are Hermitian and satisfy the commutation relation
$[J_k,J_l] = i e_{klm} J_m$, where $e_{klm}$ is the conventional
Levi-Civita symbol and the summation over $m$ being assumed. Note
that the spin-polarization vector ${\bf p}$ does not contain the
full information about the quantum state if $j \geqslant 1$.
Despite this fact, it is of great use in quantum physics and
chemistry as its components represent average spin projections
onto three orthogonal axes and are experimentally measurable.
Linear transformations of the spin polarization vector include
rotations and scaling. Rotations are attributed to the unitary
evolution, so we do not consider them in the present paper.
Physically motivated scaling of the spin polarization vector is
described by a map $\Phi: \mathcal{B}(\mathcal{H}) \mapsto
\mathcal{B}(\mathcal{H})$ of the following form:
\begin{equation}
\label{polarization-scaling-map} \Phi[X] = \frac{1}{2j+1} {\rm
tr}[X] I + \frac{3}{j(j+1)(2j+1)} \sum_{i=1}^3 \lambda_i {\rm
tr}[X J_i] J_i \, ,
\end{equation}

\noindent where $I \in \mathcal{B}(\mathcal{H})$ is the identity
operator and $\lambda_i \in \mathbb{R}$. The factors take into
account that ${\rm tr}[I] = 2j+1$ and ${\rm tr}[J_k J_l] =
\frac{1}{3} j(j+1)(2j+1) \delta_{kl}$, $\delta_{kl}$ is the
Kronecker delta. The map~\eqref{polarization-scaling-map} is
trace-preserving and unital, i.e. ${\rm tr}\left[ \Phi[X] \right]
= {\rm tr} [X]$ and $\Phi[I] = I$. Note that the
map~\eqref{polarization-scaling-map} differs in general from other
classes of unital maps~\cite{nathanson-2007,landau-1993}. Physical
meaning of Eq.~\eqref{polarization-scaling-map} is the
transformation of the spin polarization
\begin{equation}
p_i \mapsto \lambda_i p_i, \qquad i=1,2,3.
\end{equation}

In case of spin-$\frac{1}{2}$ particles,
formula~\eqref{polarization-scaling-map} transforms into a
well-known Pauli qubit map $\Phi[X] = \frac{1}{2} \left( {\rm
tr}[X] I + \sum_{i=1}^3 \lambda_i {\rm tr}[X \sigma_i] \sigma_i
\right)$, where $(\sigma_1,\sigma_2,\sigma_3)$ is the conventional
set of Pauli matrices (see,
e.g.,~\cite{nielsen-chuang,heinosaari-ziman}). The qubit
($j=\frac{1}{2}$) map $\Phi$ is known to be positive if and only
if $|\lambda_i| \leqslant 1$, completely positive if and only if
$1 \pm \lambda_3 \geqslant |\lambda_1 \pm \lambda_2|$,
entanglement breaking if and only if $|\lambda_1| + |\lambda_2| +
|\lambda_3| \leqslant 1$, 2-local-entanglement-annihilating if and
only if $\lambda_1^2+\lambda_2^2+\lambda_3^2 \leqslant 1$,
2-tensor-stable positive if and only if $1 \pm \lambda_3^2
\geqslant |\lambda_1^2 \pm
\lambda_2^2|$~\cite{filippov-rybar-ziman-2012,filippov-2014,filippov-magadov}.
Similar characterization for higher spins ($j\geqslant 1$) is
still missing, so the goal of the present paper is to analyze
analogous properties of such maps and illustrate them for qutrits
($j=1$).

The paper is organized as follows. In
Sec.~\ref{section-positivity}, we analyze positivity of the
map~\eqref{polarization-scaling-map}. In Sec.~\ref{section-cp},
the criterion of complete positivity of such a map is presented.
In Sec.~\ref{section-eb}, the entanglement breaking property is
partially characterized. In Sec.~\ref{section-2-tsp}, we review
the positivity and entanglement annihilation behaviour of the map
$\Phi\otimes\Phi$. In Sec.~\ref{section-conclusions}, brief
conclusions are presented.

\section{Positivity}
\label{section-positivity}

We will refer to an operator $R$ as positive-semidefinite and
write $R \geqslant 0$ if $\bra{\varphi} R \ket{\varphi} \geqslant
0$ for all $\ket{\varphi} \in \mathcal{H}$. A map $\Phi$ is called
positive if $\Phi[X] \geqslant 0$ for all $X \geqslant
0$~\cite{stormer-1963}.

Let us now analyze positivity of the spin polarization-scaling
map~\eqref{polarization-scaling-map}.

Since each $J_i$ is a spin projection operator with eigenvalues
$j, j-1, \ldots, -j$, eigenvalues of the operator $a_1 J_1 + a_2
J_2 + a_3 J_3$ are $|{\bf a}|\{j,j-1,\ldots,-j\}$, where $|{\bf
a}| = \sqrt{a_1^2 + a_2^2 + a_3^2}$. Therefore, the minimal
eigenvalue of $\Phi[X]$ reads
\begin{equation}
\label{minimal-eigenvalue} \frac{1}{2j+1} \left( {\rm tr}[X] -
\frac{3}{j+1} \sqrt{ \sum_{i=1}^3 \left( \lambda_i {\rm tr}[X J_i]
\right)^2} \right).
\end{equation}

\noindent Suppose $X\geqslant 0$. As $|{\rm tr}[X J_i]| \leqslant
j {\rm tr}[X]$, the minimal value of \eqref{minimal-eigenvalue} is
non-negative if $1 - \frac{3j}{j+1}\sqrt{\sum_{i=1}^3 \lambda_i^2}
\geqslant 0$. Thus, we have found sufficient condition for
positivity of the map $\Phi$.

\begin{proposition}
Spin-polarization-scaling map $\Phi$ is positive if $\sum_{i=1}^3
\lambda_i^2 \leqslant \left( \frac{j+1}{3j} \right)^2$.
\end{proposition}

The necessary condition for positivity of the
map~\eqref{polarization-scaling-map} follows from the particular
form of the positive-semidefinite operator
\begin{equation}
\label{X-positive} X = x_0 I + \sum_{i=1}^3 x_i J_i, \quad x_0
\geqslant 0, \quad x_i \in \mathbb{R}, \quad x_0 \geqslant j
\sqrt{\sum_{i=1}^3 x_i^2}.
\end{equation}

\noindent In fact, if $X$ is given by formula \eqref{X-positive},
then $\Phi[X] = x_0 I + \sum_{i=1}^3 \lambda_i x_i J_i$ and
$\Phi[X] \geqslant 0$ if and only if $x_0 \geqslant j
\sqrt{\sum_{i=1}^3 \lambda_i^2 x_i^2}$. Suppose $x_0 = j x_1$ and
$x_2 = x_3 = 0$, then $\Phi[X] \geqslant 0$ if $|\lambda_1|
\leqslant 1$. Similarly, necessary conditions
$|\lambda_2|\leqslant 1$ and $|\lambda_3|\leqslant 1$ appear for
choices $x_0 = j x_2$, $x_1 = x_3 = 0$ and $x_0 = j x_3$, $x_1 =
x_2 = 0$, respectively.

\begin{proposition}
Suppose the spin polarization-scaling map $\Phi$ is positive, then
$|\lambda_i| \leqslant 1$, $i=1,2,3$.
\end{proposition}

\section{Complete positivity}
\label{section-cp}

A linear map $\Phi$ is called completely positive if $\Phi \otimes
{\rm Id}_k$ is positive for all $k=1,2,\ldots$. Here ${\rm Id}_k$
is the identity transformation of $k$-dimensional operators
$\mathcal{B}(\mathcal{H}_k)$.

\begin{proposition}
The spin polarization-scaling map $\Phi$ is completely positive if
and only if
\begin{equation}
\label{choi-positive} I \otimes I + \frac{3}{j(j+1)} \left(
\lambda_1 J_1 \otimes J_1 - \lambda_2 J_2 \otimes J_2 + \lambda_3
J_3 \otimes J_3 \right) \geqslant 0.
\end{equation}
\end{proposition}

\begin{proof}
A linear map $\Phi:\mathcal{B}(\mathcal{H}_d) \mapsto
\mathcal{B}(\mathcal{H}_d)$ is known to be completely positive if
and only if its Choi matrix $\Omega_{\Phi} = (\Phi\otimes{\rm
Id}_d)[\ket{\psi_+}\bra{\psi_+}]$ is
positive-semidefinite~\cite{choi-1975} (see
also~\cite{pillis-1967,jamiolkowski-1972,jiang-2013,majewski-2013}),
where $\ket{\psi_+} = \frac{1}{\sqrt{d}}\sum_{i=1}^d \ket{i}
\otimes \ket{i}$ is the maximally entangled state and
$\{\ket{i}\}_{i=1}^d$ is some orthonormal basis in
$\mathcal{H}_d$.

For our construction of the Choi operator let us choose
eigenvectors of the operator $J_3$ as the basis, namely, $J_3
\ket{jm} = m \ket{jm}$ and $\ip{jm'}{jm} = \delta_{m'm}$, $d= 2j
+1$. Introduce auxiliary operators $J_{\pm} = J_1 \pm i J_2$, then
$J_{\pm} \ket{jm} = \sqrt{(j \mp 1)(j \pm m + 1)} \ket{j m \pm
1}$. Some algebra yields
\begin{eqnarray}
\label{choi-matrix} && (2j+1)^2 (\Phi \otimes {\rm
Id})[\ket{\psi_+}\bra{\psi_+}] \nonumber\\
&& = \sum_{m,m'=-j}^j \Phi[\ket{jm}\bra{jm'}] \otimes
\ket{jm}\bra{jm'}
\nonumber\\
&& = \sum_{m=-j}^j \left( I + \frac{3m}{j(j+1)} \lambda_3 J_3
\right) \otimes \ket{jm}\bra{jm} \nonumber\\
&& \quad + \sum_{m=-j}^j \frac{3}{2j(j+1)} \sqrt{(j-m)(j+m+1)}
 \nonumber\\
 && \qquad\qquad \times (\lambda_1 J_1 - i \lambda_2 J_2) \otimes \ket{m} \bra{m+1}
\nonumber\\
&& \quad + \sum_{m=-j}^j \frac{3}{2j(j+1)} \sqrt{(j+m)(j-m+1)}
 \nonumber\\
 && \qquad\qquad \times (\lambda_1
J_1 + i \lambda_2 J_2) \otimes \ket{m} \bra{m-1} \nonumber\\
&& = I \otimes I + \frac{3}{j(j+1)} \left( \lambda_1 J_1 \otimes
J_1 - \lambda_2 J_2 \otimes J_2 + \lambda_3 J_3 \otimes J_3
\right). \nonumber\\
\end{eqnarray}

\noindent Thus, $\Phi$ is completely positive if and only if the
operator \eqref{choi-matrix} is positive-semidefinite.
\end{proof}

\begin{example}
If $j=\frac{1}{2}$, then angular momentum operators are given by
matrices
\begin{equation}
J_1 = \frac{1}{2} \left(%
\begin{array}{cc}
  0 & 1 \\
  1 & 0 \\
\end{array}%
\right), \quad J_2 = \frac{1}{2} \left(%
\begin{array}{cc}
  0 & -i \\
  i & 0 \\
\end{array}%
\right), \quad J_3 = \frac{1}{2} \left(%
\begin{array}{cc}
  1 & 0 \\
  0 & -1 \\
\end{array}%
\right)
\end{equation}

\noindent in the basis $\{\ket{\frac{1}{2} \frac{1}{2}},
\ket{\frac{1}{2} -\frac{1}{2}}\}$. The
condition~\eqref{choi-positive} reduces to $1 \pm \lambda_3
\geqslant | \lambda_1 \pm \lambda_2 |$. Geometrically, these
inequalities correspond a tetrahedron with vertices $(1,1,1)$,
$(1,-1,-1)$, $(-1,1,-1)$, and $(-1,-1,1)$ in the parameter space
$(\lambda_1,\lambda_2,\lambda_3)$~\cite{ruskai-2002}.
\end{example}

\begin{figure*}
\includegraphics[width=12cm]{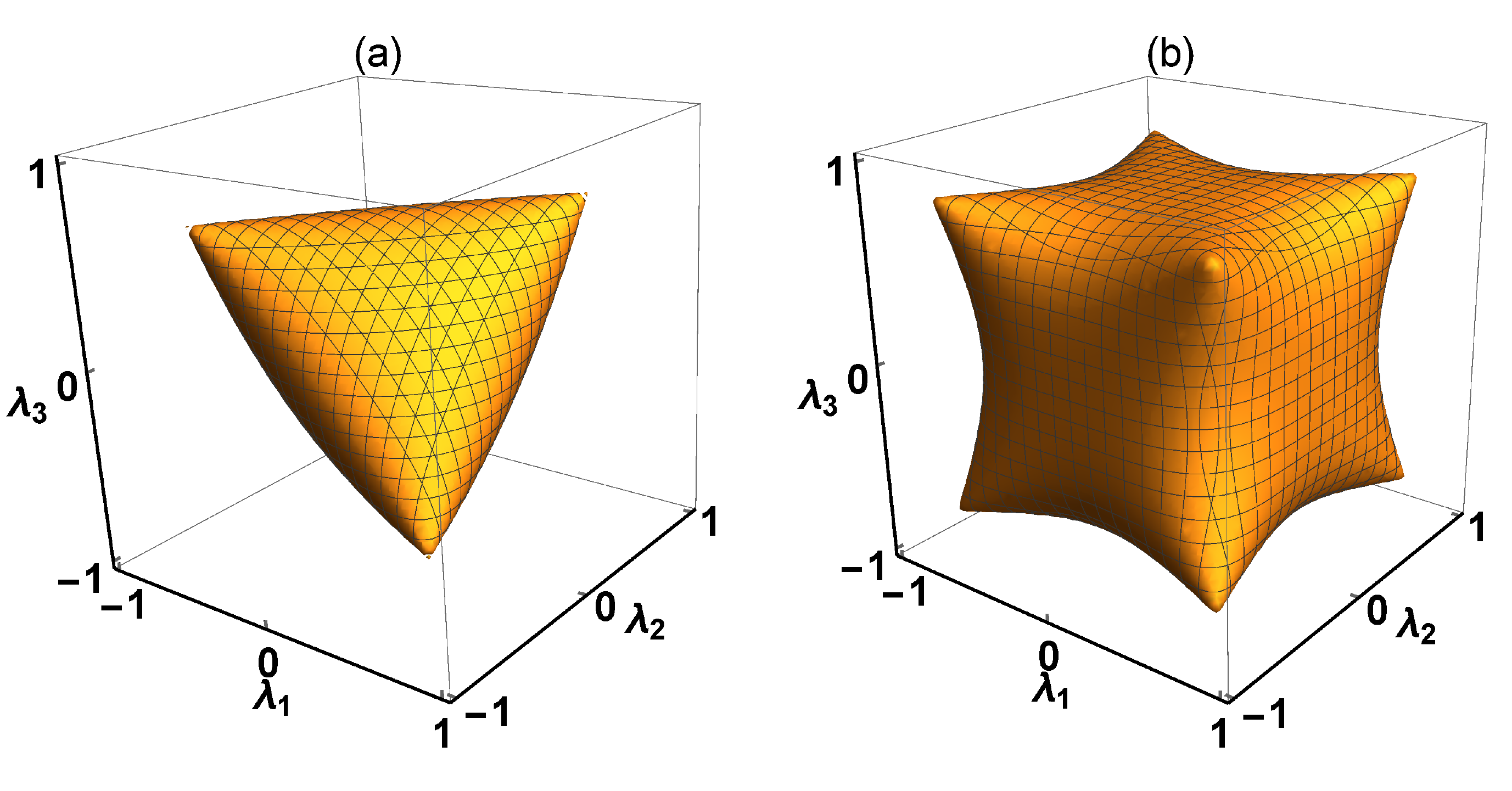}
\caption{\label{figure} Region of parameters
$(\lambda_1,\lambda_2,\lambda_3)$ for which the spin-$1$
polarization map $\Phi$ is completely positive (a), the map
$\Phi^2$ is completely positive (b).}
\end{figure*}

\begin{example}
If $j = 1$, then angular momentum operators are given by matrices
\begin{eqnarray}
&& J_1 = \frac{1}{\sqrt{2}} \left(%
\begin{array}{ccc}
  0 & 1 & 0 \\
  1 & 0 & 1 \\
  0 & 1 & 0 \\
\end{array}%
\right), \quad \ J_2 = \frac{1}{\sqrt{2}} \left(%
\begin{array}{ccc}
  0 & -i & 0 \\
  i & 0 & -i \\
  0 & i & 0 \\
\end{array}%
\right), \nonumber\\
&& J_3 = \left(%
\begin{array}{ccc}
  1 & 0 & 0 \\
  0 & 0 & 0 \\
  0 & 0 & -1 \\
\end{array}%
\right)
\end{eqnarray}

\noindent in the basis $\{\ket{11}, \ket{10}, \ket{1-1}\}$. The
condition~\eqref{choi-positive} reduces to $4-9(\lambda_1^2 +
\lambda_2^2 + \lambda_3^2) + 27 \lambda_1 \lambda_2 \lambda_3
\geqslant 0$, $|\lambda_i| \leqslant \frac{2}{3}$, $i=1,2,3$.
Geometrical figure corresponding to such inequalities is depicted
in Fig.~\ref{figure}(a).
\end{example}

\section{Entanglement breaking}
\label{section-eb}

A positive-semidefinite operator $R \in \mathcal{B}(\mathcal{H}_1)
\otimes \mathcal{B}(\mathcal{H}_2)$ is called separable
(non-entangled) if there exist positive-semidefinite operators
$R_1^{(k)}$ and $R_2^{(k)}$ such that $R = \sum_k R_1^{(k)}
\otimes R_2^{(k)}$~\cite{werner-1989,horodecki-2009}. A linear map
$\Phi: \mathcal{H}_1 \mapsto \mathcal{H}_1$ is called entanglement
breaking if $(\Phi \otimes {\rm Id}) [\varrho]$ is separable for
all $\varrho \in \mathcal{S}(\mathcal{H}_1 \otimes \mathcal{H}_2)$
and identity transformation ${\rm Id}: \mathcal{H}_2 \mapsto
\mathcal{H}_2$~\cite{holevo-1998,king-2002,shor-2002,ruskai-2003,horodecki-2003}.
The well-known result is that $\Phi$ is entanglement breaking if
and only if the Choi matrix $\Omega_{\Phi}$ is separable.

The necessary condition for separability of $\Omega_{\Phi}$ is
that $\Omega_{\Phi}^{\Gamma} \geqslant 0$, where $X^{\Gamma} =
\sum_{j,j'} I \otimes \ket{j'}\bra{j} X I \otimes \ket{j'}\bra{j}$
is the partially transposed operator, $X\in
\mathcal{B}(\mathcal{H}_1) \otimes
\mathcal{B}(\mathcal{H}_2)$~\cite{peres-1996,horodecki-1996}.
Applying such a condition to the Choi matrix~\eqref{choi-matrix}
and taking into account that in conventional basis
$\{jm\}_{m=-j}^j$ the matrices $J_x^{\top} = J_x$, $J_y^{\top} =
-J_y$, $J_z^{\top} = J_z$, we obtain the following result.

\begin{proposition}
Suppose the spin polarization-scaling map $\Phi$ is completely
positive and entanglement breaking, then
\begin{equation} \label{choi-PPT}
I \otimes I + \frac{3}{j(j+1)} \left( \lambda_1 J_1 \otimes J_1
\pm \lambda_2 J_2 \otimes J_2 + \lambda_3 J_3 \otimes J_3 \right)
\geqslant 0.
\end{equation}
\end{proposition}

Note that the requirement~\eqref{choi-PPT} is sufficient for the
channel $\Phi$ to be entanglement binding~\cite{horodecki-2000}.

\begin{example}
If $j=\frac{1}{2}$, then Eq.~\eqref{choi-PPT} is equivalent to
$|\lambda_1| + |\lambda_2| + |\lambda_3| \leqslant 1$.
\end{example}

\begin{example}
If $j=1$, then Eq.~\eqref{choi-PPT} is equivalent to
$4-9(\lambda_1^2 + \lambda_2^2 + \lambda_3^2) \pm 27 \lambda_1
\lambda_2 \lambda_3 \geqslant 0$.
\end{example}

\section{2-tensor-stable properties}
\label{section-2-tsp}

Some properties of linear maps do not change under tensoring the
map with itself, for instance, $\Phi \otimes \Phi$ is completely
positive if and only if $\Phi$ is completely positive. Similarly,
$\Phi \otimes \Phi$ is entanglement breaking if and only if $\Phi$
is entanglement breaking (see,
e.g.,~\cite{filippov-melnikov-ziman-2013}). However, other
properties of a map can change drastically under tensor power. For
example, the map $\Phi \otimes \Phi$ can be non-positive even if
$\Phi$ is positive~\cite{filippov-magadov}.

A linear map $\Phi$ is called 2-tensor-stable positive if $\Phi
\otimes \Phi$ is positive~\cite{muller-hermes-2016}.

\begin{example}
\label{example-qubit} It is shown in Ref.~\cite{filippov-magadov}
that the spin polarization-scaling map $\Phi$ given by
Eq.~\eqref{polarization-scaling-map} for qubits ($j=\frac{1}{2}$)
is 2-tensor-stable positive if and only if $\Phi^2$ is completely
positive, i.e. $1 \pm \lambda_3^2 \geqslant |\lambda_1^2 \pm
\lambda_2^2|$.
\end{example}

For higher spins ($j \geqslant 1$) the result of
Example~\ref{example-qubit} can be extended as follows.

\begin{proposition}
\label{proposition-2-tsp} If the spin polarization-scaling map
$\Phi$ is 2-tensor-stable positive, then $\Phi^2$ is completely
positive.
\end{proposition}

\begin{proof}
Consider a positive-semidefinite operator
$\ket{\psi_+}\bra{\psi_+}$, where $\ket{\psi_+} = (2j+1)^{-1/2}
\sum_{m=-j}^j \ket{jm} \otimes \ket{jm}$. The action of the
positive map $\Phi \otimes \Phi$ on such an operator reads
\begin{eqnarray}
&& 0 \leqslant (2j+1)^2 (\Phi \otimes
\Phi)[\ket{\psi_+}\bra{\psi_+}] \nonumber\\
&& = (2j+1)^2 ({\rm Id} \otimes \Phi) \circ (\Phi \otimes {\rm
Id})
[\ket{\psi_+}\bra{\psi_+}] \nonumber\\
&& = ({\rm Id} \otimes \Phi) \bigg[ I \otimes I + \frac{3}{j(j+1)}
\nonumber\\
&& \qquad\qquad\quad \times \left( \lambda_1 J_1 \otimes J_1 -
\lambda_2 J_2 \otimes J_2 + \lambda_3 J_3 \otimes J_3 \right)
\bigg]
\nonumber\\
&& = I \otimes I + \frac{3}{j(j+1)} \left( \lambda_1^2 J_1 \otimes
J_1 - \lambda_2^2 J_2 \otimes J_2 +
\lambda_3^2 J_3 \otimes J_3 \right) \nonumber\\
&& = (2j+1)^2 (\Phi^2 \otimes {\rm Id})[\ket{\psi_+}\bra{\psi_+}],
\end{eqnarray}

\noindent i.e. the Choi matrix $\Omega_{\Phi^2}$ is
positive-semidefinite and $\Phi^2$ is completely positive.
\end{proof}

In contrast to the case $j=\frac{1}{2}$, for higher spins ($j
\geqslant 1$) Proposition~\ref{proposition-2-tsp} provides the
necessary condition only. For instance, it is not hard to see that
for $j=1$ there exists a spin polarization-scaling map $\Phi$ such
that $\Phi^2$ is completely positive but
$(\Phi\otimes\Phi)[\ket{\varphi}\bra{\varphi}] \ngeqslant 0$ for
the Schmidt-rank-2 state $\ket{\varphi} = \frac{1}{\sqrt{2}}\left(
\ket{11} + \ket{1-1} \right)$.

In the case $j=1$, the map $\Phi^2$ is completely positive if and
only if $4-9(\lambda_1^4 + \lambda_2^4 + \lambda_3^4) + 27
\lambda_1^2 \lambda_2^2 \lambda_3^2 \geqslant 0$, which is
depicted in Fig.~\ref{figure}(b).

A linear map $\Phi: \mathcal{B}(\mathcal{H}) \mapsto
\mathcal{B}(\mathcal{H})$ is called 2-locally entanglement
annihilating~\cite{moravcikova-ziman-2010,filippov-rybar-ziman-2012,filippov-ziman-2013,filippov-ziman-2014}
if $(\Phi\otimes\Phi)[\ket{\psi}\bra{\psi}]$ is separable for all
$\ket{\psi} \in \mathcal{H} \otimes \mathcal{H}$. The same line of
reasoning as for 2-tensor-stable positive maps leads to the
following result.

\begin{proposition}
\label{proposition-2-lea} If the spin polarization-scaling map
$\Phi$ is 2-locally entanglement annihilating, then $\Phi^2$ is
entanglement breaking.
\end{proposition}

\section{Conclusions}
\label{section-conclusions}

We have considered the physically motivated sets of operator maps
for spin systems. The physical meaning of such maps is the
degradation of spin polarization with scaling parameters
$\lambda_1,\lambda_2,\lambda_3$ along the axes $x,y,z$,
respectively. We have found conditions (necessary, or sufficient,
or both) under which the spin polarization-scaling map is
positive, completely positive, entanglement breaking,
2-tensor-stable positive, 2-locally entanglement annihilating.
These results can be of use in the analysis of data, where only
spin polarization degrees of freedom are available. The crucial
difference between the cases of spin-$\frac{1}{2}$ and spin-$1$
particles is illustrated in a series of examples.

\begin{acknowledgements}
The study is supported by Russian Science Foundation under project
No. 16-11-00084.
\end{acknowledgements}

\end{document}